\documentclass[11pt]{article}

\usepackage{sn-preamble}
\usepackage{thm-restate}
\usepackage[margin=1in]{geometry}
\usepackage{tikz}
\usepackage{verbatim}
\usepackage{cancel}
\usepackage{titling}
\usepackage{multirow}
\usepackage{cancel}
\usepackage{quantikz}

\newcommand{\red}[1]{{{#1}}}

\newcommand{\ignore}[1]{}
\definecolor{cb-salmon-pink}{RGB}{255, 182, 119}

\newcommand{\JuntaTest}{\textsc{Junta-Tester}}
\newcommand{\JuntaUniformity}{\textsc{Junta-Uniformity-Test}}
\newcommand{\UniformityTest}{\textsc{Uniformity-Test}}
\newcommand{\ConsistentJunta}{\textsc{Consistent-Junta-Checker}}
\newcommand{\MLE}{\calA_{\mathrm{MLE}}}
\newcommand{\dtv}{\mathrm{d}_{\mathrm{TV}}}

\title{Testing Junta Truncation}
\author{
 William He\thanks{Carnegie Mellon University. Email: \url{wrhe@cs.cmu.edu}.}
 \and 
 Shivam Nadimpalli\thanks{Columbia University. Email: \url{sn2855@columbia.edu}.} 
\vspace{0.75em}}
\date{\small\today}

\begin{document}

\pagenumbering{gobble}
\hypersetup{linkcolor={black}}

\maketitle
\begin{abstract}
We consider the basic statistical problem of detecting  truncation of the uniform distribution on the Boolean hypercube by juntas. More concretely, we give upper and lower bounds on the problem of distinguishing between i.i.d. sample access to either (a) the uniform distribution over $\zo^n$, or (b) the uniform distribution over $\zo^n$ conditioned on the satisfying assignments of a $k$-junta $f\isazofunc$.

We show that (up to constant factors) $\min\{2^k + \log{n\choose k}, {2^{k/2}\log^{1/2}{n\choose k}}\}$ samples suffice for this task and also show that a $\log{n\choose k}$ dependence on sample complexity is unavoidable. Our results suggest that testing junta truncation requires learning the set of relevant variables of the junta. 
\end{abstract}

\hypersetup{linkcolor={purple}}

\newpage
\pagenumbering{arabic}

\section{Introduction}
\label{sec:intro}

Estimation via samples acquired from a \emph{truncated} distribution---i.e. a distribution that has undergone some form of conditioning---is a classic statistical challenge dating back to the works of Bernoulli~\cite{Bernoulli60}, Galton~\cite{Galton97},  and Pearson~\cite{Pearson02}. 
The premise of inference and learning from truncated distributions suggests an arguably more basic statistical task, namely that of \emph{detecting} if a distribution has been truncated in the first place. We instantiate this question in perhaps its simplest high-dimensional setting by testing for truncations of the uniform distribution on the Boolean hypercube $\zo^n$ by functions of few variables.

The concept of a ``function of few variables'' can be made precise via the notion of \emph{juntas}: 
A function $f\isazofunc$ is a \emph{$k$-junta} if there is a set of coordinates $\{i_1, \ldots, i_k\} \sse [n]$ such that the value of $f(x)$ is completely determined by $(x_{i_1}, \ldots, x_{i_k})$. Starting with the work of Blum~\cite{Blu94}, the problem of learning juntas (which abstracts the problem of learning in the presence of irrelevant features) as well as the related problem of \emph{testing} juntas (i.e.~the problem of distinguishing whether a function is a $k$-junta or is ``far'' from every $k$-junta) have been---and continue to be---the subject of intensive study in theoretical computer science~\cite{BL97,MOS04,FKRSS04,Bla08,Val15,BCELR18,DMN19,ITW21,PRW22,CNY23,CP23,LYT-comp}.

We now turn to a more precise formulation of the question we consider in this paper: Given i.i.d.~sample access to an unknown distribution $\calD$ on $\zo^n$, distinguish with high probability between (a) $\calD$ being the uniform distribution on $\zo^n$, and (b) the uniform distribution on the satisfying assignments\footnote{We say that $x$ is a \emph{satisfying assignment} of a Boolean function $f\isazofunc$ if $f(x) = 1$.} of some $k$-junta (assuming there are not too many satisfying assignments). While the the learnability and testability of the truncated distributions we consider (namely those in Item~(b) above) are well studied problems~\cite{ABR16,CJLW21}, the question we consider is an arguably more basic question and adds to a nascent line of inquiry on testing distribution truncation~\cite{DNS23-convex,DLNS23-ptf}; we defer a detailed discussion of related work to \Cref{subsec:related}.\footnote{\red{\cite{ABR16}~and~\cite{CJLW21} consider the broader class of distributions on $\zo^n$ whose density function is a $k$-junta; our upper bound partially (in particular,~\Cref{alg:uniformity-junta-checker}) holds for this broader class of distributions as well---see the discussion following \Cref{thm:ub} for more on this.}}

\paragraph{Notation.} Before formally stating our results, we first introduce some notation. We will denote the uniform distribution over the $n$-dimensional Boolean hypercube $\zo^n$ by $\calU_n$. For a function $f\isazofunc$, we define its \emph{volume} $\vol(f)$ as 
\[\vol(f) := \Ex_{\bx\sim\zo^n}\sbra{f(\bx)}.\]
Finally, we introduce notation for the class of distributions we consider: 
\begin{equation} \label{eq:junta-def}
    \calJ_n(k,\eps) := \cbra{\calU_n|_{f^{-1}(1)} ~\text{for a $k$-junta}~f\isazofunc~\text{with}~ \vol(f) \leq 1-\epsilon},
\end{equation}
namely the collection of uniform distributions on satisfying assignments of $k$-juntas with volume at most $1-\eps$. We note that $\calU_n$ can be viewed as a truncation of $\calU_n$ by the constant-$1$ function which is a $0$-junta; as such, the assumption on the volume of the junta not being too large is necessary.

\subsection{Our Results} 
\label{subsec:results}

We obtain both upper and lower bounds on the sample complexity of distinguishing truncations of the uniform distribution on $\zo^n$ by juntas.

\subsubsection{Upper Bound}
\label{subsubsec:ub}

We obtain the following upper bound:

\begin{theorem} \label{thm:ub}
    There exists an algorithm \JuntaTest~with the following performance guarantee: Given i.i.d. sample access to an unknown distribution $\calD$, the algorithm draws $O\pbra{\min\{T_1, T_2\}}$ samples where 
    \[
    T_1 := 2^k + \log{n\choose k} \qquad\text{and}\qquad T_2 := {2^{k/2}\log^{1/2}{n\choose k}} + \log{n\choose k},
    \] 
    and has the following performance guarantee:
    \begin{itemize}
        \item If $\calD = \calU_n$, then $\calA$ outputs ``un-truncated'' with probability at least $99/100$; and
        \item If $\calD \in \calJ_n(k, \eps)$, then $\calA$ outputs ``truncated'' with probability at least $99/100$.
    \end{itemize}
\end{theorem}

Our upper bound, which we prove in \Cref{sec:junta-ub}, is based on two distinct and strikingly simple algorithms for testing junta truncation: 
\begin{enumerate}
    \item The first algorithm, \ConsistentJunta~(cf. \Cref{alg:consistent-junta-checker}), assumes the distribution is truncated as long as there is a junta consistent with the samples. As such, it can be viewed as a ``consistent hypothesis checker''~\cite{kearns1994introduction} and has sample complexity $T_1$. Note that if the samples are from a truncated distribution, then this algorithm learns the truncating junta.

    \item The second algorithm, \JuntaUniformity~(cf.~\Cref{alg:uniformity-junta-checker}), relies on the easy observation that in the truncated case, the distribution over the relevant variables in the junta will be far from uniform (thanks to the assumption on the volume on the function, cf.~\Cref{eq:junta-def} and the following discussion). The algorithm thus draws a sufficiently large number of samples and iterates over all subsets of $k$ variables, and relies on the uniformity test of Diakonikolas et al.~\cite{DGPP} to distinguish between the untruncated and truncated distributions.
\end{enumerate}

Note that \ConsistentJunta~is computationally inefficient as it iterates over all $\Theta(2^{2^k}{n\choose k})$ $k$-juntas over $\zo^n$ whereas \JuntaUniformity~runs in time $O(T_2)$ due to the uniformity test subroutine of~\cite{DGPP}.

\begin{remark}
    The combined use of the two algorithms points to an intriguing sample-versus-time complexity trade-off with the parameter $k$. In particular, note that when 
    \[
    2^k \ll \log{n\choose k} \qquad\text{i.e.}\qquad
    k \ll \red{\log\log n},
    \]
    we have that computationally {inefficient}~\ConsistentJunta~has lower sample complexity than the computationally efficient~\JuntaUniformity. 
\end{remark}

\subsubsection{Lower Bound}
\label{subsubsec:lb}

We now turn to our lower bound on the sample complexity of testing junta truncation:

\begin{theorem} \label{thm:lb}
    Let $\calA$ be any algorithm with the same performance guarantee as in \Cref{thm:ub}. Then $\calA$ must draw $\Omega(\log{\binom{n}{k}})$ samples from $\calD$.
\end{theorem}

Our lower bound, which we prove in \Cref{sec:lb}, is obtained by proving that $\log{n\choose k}$ samples are required to distinguish between $\calU_n$ and $\calU_n$ conditioned on the satisfying assignments of a parity function on a uniformly random set of coordinates of size $k$; our proof is linear algebraic and reduces to proving an elementary anti-concentration property of a coding-theoretic random variable. We remark that the lower bound in \Cref{thm:lb} matches the upper bounds in \Cref{thm:ub} up to constant factors for $k = \Theta(1)$ and $k = \red{\Theta(
\log\log n)}$. Our lower bound also suggests that in order to test truncation by a $k$-junta, \emph{learning} the set of relevant variables of the junta is in a sense necessary.

\subsection{Related Work} 
\label{subsec:related}

Estimation and inference over truncated distributions have been increasingly studied by the theoretical computer science community in recent years~\cite{DGTZcolt19,FKTcolt20, FKTcolt20,DKTZcolt21}; we refer the reader to Section~1.3 of \cite{DNS23-convex} for more on this. The question we consider can be viewed as a specific instance of the following broader question: 

\begin{question} \label{question:potato}
{Given independently drawn samples from some unknown distribution, determine whether the data was drawn from a known underlying probability distribution  ${\cal P}$, versus from ${\cal P}$ conditioned on some unknown \emph{truncation set $S$} of measure bounded away from $1$.}
\end{question}

The work of Rubinfeld and Servedio~\cite{RubinfeldS09} considers the problem of testing for truncations of the uniform distribution on $\zo^n$ by {monotone} Boolean functions, and the work of De, Nadimpalli, and Servedio~\cite{DNS23-convex} considers an instant of \Cref{question:potato} where $\calP$ is the $n$-dimensional standard Gaussian distribution and the truncating set $S$ can be an arbitrary {convex} subset of $\R^n$. We note that both \cite{RubinfeldS09} and \cite{DNS23-convex} obtain a striking separation between the sample complexities of learning and testing truncation by their respective classes of truncating sets: For both monotone functions and convex sets, the best known learning algorithms require essentially $n^{O(\sqrt{n})}$ samples, whereas testing truncation in both cases requires $O(n)$ samples.\footnote{This can be viewed as yet another case of the ``emerging analogy between monotone Boolean functions and convex sets''~\cite{DNS22}.} 

In addition to \cite{RubinfeldS09,DNS23-convex}, a forthcoming manuscript of De, Li, Nadimpalli, and Servedio~\cite{DLNS23-ptf} considers a broader instantiation of \Cref{question:potato} wherein $\calP$ is taken to be an arbitrary ``{hypercontractive}'' distribution (cf. Chapters~9 and 10 of~\cite{ODonnell2014}) and the truncating set $S$ can be the satisfying assignments of a low-degree polynomial threshold function. 

\section{Preliminaries}
\label{sec:preliminaries}

All probabilities and expectations will be with respect to the uniform distribution, unless otherwise indicated. We use boldfaced characters such as $\bw, \bx,$ and $\bV$ to denote random variables (which may be real-valued, vector-valued or set-valued; the intended type will be clear from the context).
We write $\bx \sim \calD$ to indicate that the random variable $\bx$ is distributed according to the probability distribution $\calD$. We write $\dtv(\calD_1,\calD_2)$ to denote the \emph{total variation distance} or \emph{statistical distance} between the distributions $\calD_1$ and $\calD_2$.

We will sometimes identify the Boolean hypercube $\zo^n$ with the vector space $\F_2^n$. For $x\in\zo^n$ we will write $|x|$ to denote the {Hamming weight} of $x$, i.e. $|x| := \sumi x_1$; viewed as an element of $\F_2^n$, we have $|x| = |\{i\in[n] : x_i\neq 0\}|$. Given a subset $S\sse[n]$ and $x\in\zo^n$, we will write 
\[x_S := (x_{i})_{i\in S}\]
and will view $x_S$ as an element of $\zo^S$ in the natural way.
\section{Upper Bounds}
\label{sec:junta-ub}

In this section, we establish our upper bound on testing junta truncation (cf. \Cref{thm:ub}). As discussed earlier, we give two distinct algorithms for detecting truncation by juntas: In \Cref{subsec:testing-via-learning}, we show that it is possible to test junta truncation using $O(2^k+\log\binom{n}{k})$ samples via learning. Our second algorithm, which we present in \Cref{subsec:testing-via-collisions}, relies on a uniformity test due to Diakonikolas et al.~\cite{DGPP} and requires $O({2^{k/2}\log^{1/2}{n\choose k}} + \log{n\choose k})$ samples.

\subsection{Testing Truncation via a Consistent Hypothesis Checker}
\label{subsec:testing-via-learning}

The truncation tester we present in this section can be viewed as a ``consistent hypothesis finder'' (see for example~\cite{kearns1994introduction}) where the algorithm simply checks if there exists a junta that is consistent with all the samples and outputs ``truncated'' if this is indeed the case. We note that this algorithm is computationally inefficient, and has runtime $O(2^{2^k}{n\choose k})$ as it must iterate over all $k$-juntas over $\zo^n$.

\begin{algorithm}[t]
    \addtolength\linewidth{-5.5ex}
    \vspace{0.5em}
    \textbf{Input:} $\calD \in \{\calU_n\} \sqcup \calJ_n(k,\eps)$\\[0.25em]
    \textbf{Output:} ``Un-truncated'' or ``truncated''

    \

    $\ConsistentJunta(\calD)$:
    \begin{enumerate}
        \item Set 
        \[T := 100\pbra{2^k + \log{n \choose k}}\log^{-1}\pbra{\frac{1}{1-\eps}}\]
        and draw $\x{1}, \ldots, \x{T} \sim \calD$.
        \item If there exists a $k$-junta $f\isazofunc$ with $\Vol(f) \leq 1-\eps$ such that for all $i\in[T]$ we have $f(\x{i}) = 1$, then ouput ``truncated;'' otherwise, output ``un-truncated.''
    \end{enumerate}
    
    \caption{Testing junta truncation via a consistent hypothesis checker.}
    \label{alg:consistent-junta-checker}
\end{algorithm}

\begin{proposition} \label{prop:consistent-junta-ub}
    The algorithm \ConsistentJunta~(cf. \Cref{alg:consistent-junta-checker}) has the following performance guarantee:
    \begin{itemize}
        \item If $\calD = \calU_n$, then it outputs ``un-truncated'' with probability at least $99/100$; and
        \item If $\calD \in \calJ_n(k,\eps)$, then it outputs ``truncated'' with probability $1$.
    \end{itemize}
\end{proposition}

\begin{proof}
    Note that if $\calD\in\calJ_n(k,\eps)$, then clearly the algorithm never outputs ``un-truncated.'' It therefore suffices to show that when $\calD = \calU_n$, the algorithm outputs ``truncated'' with probability at most $1/100$. 
    Let $\bX$ be the random variable that counts the number of $k$-juntas with volume at most $(1-\eps)$ that are consistent with the sample set $(\x{1},\ldots,\x{T})$, i.e.
    \[\bX := \sum_{\calD' \in\calJ_n(k, \eps)} \mathbf{1}\cbra{\x{i} \in \supp(\calD') ~\text{for all}~i\in[T]}\]
    where $\supp(\calD')$ denotes the support of the distribution $\calD'$.
    We then have 
    \begin{align*}
        \Ex_{\calU_n}\sbra{\bX} &= \sum_{\calD'\in\calJ_n(k,\eps)} \Prx_{\calU_n}\sbra{\x{i} \in \supp(\calD') ~\text{for all}~i\in[T]}\\
        &\leq \sum_{\calD'\in\calJ_n(k,\eps)} \pbra{1-\eps}^T\\
        &\leq 2^{2^k}{n\choose k}\cdot (1-\eps)^T \\ 
        &= 2^{2^k}{n\choose k} \cdot \pbra{\frac{1}{2}}^{100\pbra{2^k + \log{n\choose k}}}\\
        &\leq \frac{1}{100},
    \end{align*}
    where the first inequality follows from the fact that $\mathrm{d}_{\mathrm{TV}}(\calU_n, \calD') \geq \eps$, and the second inequality follows from the fact that there are at most $2^{2^k}{n\choose k}$ many $k$-juntas over $\zo^n$.
    Using Markov's inequality, we thus have that 
    \[\Prx_{\x{1}, \ldots, \x{T}\sim\calU_n}\sbra{\text{\Cref{alg:consistent-junta-checker} outputs ``truncated''}} = \Prx_{\calU_n}\sbra{\bX \geq 1} \leq \frac{1}{100},\]
    completing the proof.
\end{proof}

\subsection{Testing Truncation via Uniformity Testing}
\label{subsec:testing-via-collisions}

We now give a different algorithm for testing junta truncation that relies on the uniformity test of Diakonikolas et al.~\cite{DGPP}. The algorithm relies on the simple observation that if the input distribution has been truncated by a junta on the variables in some set $S$, then the resulting distribution on $\zo^S$ is far from uniform. We first formally state the guarantee of the uniformity test of~Diakonikolas et al.~\cite{DGPP}:

\begin{proposition}[Theorem~2 of~\cite{DGPP}] \label{prop:uniformity-test}
    Given $\eps, \delta > 0$, there exists an algorithm, \UniformityTest$(\eps,\delta)$, which given i.i.d. sample access to a distribution $\calD$ over $[m]$, draws 
    \[T := \Theta\pbra{\frac{1}{\eps^2}\pbra{\sqrt{m\log\pbra{\frac{1}{\delta}}} + \log\pbra{\frac{1}{\delta}}}}\]
    samples from $\calD$, does an $O(T)$-time computation, and has the following performance guarantee:
    \begin{itemize}
        \item If $\calD = \calU_{[m]}$ where $\calU_{[m]}$ is the uniform distribution over $[m]$, then it outputs ``un-truncated'' with probability at least $1-\delta$; and 
        \item If $\mathrm{d}_{\mathrm{TV}}(\calU_{[m]}, \calD) \geq \eps$, then it outputs ``truncated'' with probability at least $1-\delta$.
    \end{itemize}
\end{proposition}

\begin{algorithm}[t]
    \addtolength\linewidth{-5.5ex}
    \vspace{0.5em}
    \textbf{Input:} $\calD \in \{\calU_n\} \sqcup \calJ_n(k,\eps)$\\[0.25em]
    \textbf{Output:} ``Un-truncated'' or ``truncated''
    
    \
    
    $\JuntaUniformity(\calD)$:
    \begin{enumerate}
        \item Set 
        \[T := \Theta\pbra{\frac{1}{\eps^2}\pbra{2^{k/2}\cdot{\log^{1/2}\binom{n}{k}} + \log{n\choose k}}}\qquad\text{and}\qquad \delta := \Theta\pbra{\frac{1}{{n\choose k}}},\]
        and draw $T$ samples $\x{1}, \ldots, \x{T} \sim \calD$.
        \item For each $S\subseteq[n]$ with $|S|=k$: 
        \begin{enumerate}
            \item Run $\UniformityTest(\eps,\delta)$ (cf.~\Cref{prop:uniformity-test}) on samples $(\x{1}_S,\ldots, \x{T}_S)$ identifying $[m] \equiv \zo^S$. 
            \item If \UniformityTest~outputs ``truncated,'' halt and output ``truncated.''
        \end{enumerate}
        \item Otherwise, output ``un-truncated.''
    \end{enumerate}
    
    \caption{Testing junta truncation via uniformity testing.}
    \label{alg:uniformity-junta-checker}
\end{algorithm}

With this in hand, the analysis of \Cref{alg:uniformity-junta-checker} is rather straightforward:

\begin{proposition} \label{prop:junta-uniformity}
    The algorithm \JuntaUniformity~(cf.~\Cref{alg:uniformity-junta-checker}) has the following performance guarantee:
    \begin{itemize}
        \item If $\calD = \calU_n$, then it outputs ``un-truncated'' with probability at least $99/100$; and 
        \item If $\calD\in\calJ_n(k,\eps)$, then it outputs ``truncated'' with probability at least $99/100$.
    \end{itemize}
\end{proposition}

\begin{proof}    
    Note that if $\calD = \calU_n$, then for each $S\sse[n]$ with $|S|=k$, we have that $(\x{1}_S,\dots,\x{T}_S)$ is distributed according to the uniform distribution on $\zo^S$. 
    It thus follows by the guarantee of \UniformityTest~(cf. \Cref{prop:uniformity-test}) that we output ``truncated'' with probability at most $\delta$ with $\delta$ as defined in Item~1 of \Cref{alg:uniformity-junta-checker}. Union bounding over all ${n\choose k}$ sets $S\sse[n]$ of size $k$, we get (for an appropriate choice of constant hidden by the $\Theta(\cdot)$ when setting $\delta$) that the \JuntaUniformity~outputs ``truncated'' with probability at most $1/100$.

    On the other hand, suppose $\calD \in \calJ_n(k,\eps)$; in particular, suppose $\calD = \calU_n|_{f^{-1}(1)}$ where the function $f\isazofunc$ is a $k$-junta on the variables in $S^\ast \sse[n]$ with $|f^{-1}(1)|\leq (1-\eps)2^n$. In this case, note that the $(\x{1}_{S^\ast}, \ldots, \x{T}_{S^\ast})$ is drawn from a distribution on $\zo^{S^{\ast}}$ which has variation distance at least $\eps$ from the uniform distribution on $\zo^{S^{\ast}}$, and so \UniformityTest~will output ``truncated'' with probability at least $1-\delta \gg 99/100$ (once again, for an appropriate choice of constant hidden by the $\Theta(\cdot)$ when setting $\delta$).
\end{proof}

\section{Lower Bound}
\label{sec:lb}

In this section, we establish \Cref{thm:lb} by showing that any algorithm with the performance guarantee as in \Cref{thm:ub} must draw $\Omega(\log{n \choose k})$ samples from $\calD$.
We obtain our lower bound by showing that it is hard to distinguish between the uniform distribution $\calU_n$ and the uniform distribution truncated by a (negated) \emph{parity} function. Throughout this section, it will be convenient for us to identify $\zo^n$ with $\F_2^n$.

\begin{definition} \label{def:parity}
Given a subset $S\sse[n]$, we write $\chi_S:\F_2^n\to\F_2$ for the (negated) \emph{parity} {function on $S$} defined by 
\[\chi_S(x) := 1 - \sum_{i\in S} x_i.\]
In particular, note that each $\chi_S$ is a $|S|$-junta. Let $\calP_n(k)$ be the collection of truncations of the uniform distribution over $\zo^n$ by parities on $k$ variables, i.e. 
\[\calP_n(k) := \cbra{\calU_n|_{\chi_S^{-1}(1)} : |S| = k}.\]
\end{definition}

\Cref{thm:lb} follows immediately from the following:

\begin{proposition} \label{prop:parity-lb}
    Let $\calA$ be any algorithm which given i.i.d. sample access to a distribution $\calD\in\{\calU_n\}\cup\calP_n(k)$, has the following performance guarantee: 
    \begin{itemize}
        \item If $\calD = \calU_n$, then $\calA$ outputs ``un-truncated'' with probability $99/100$; and
        \item If $\calD \in \calP_n(k)$, then $\calA$ outputs ``truncated'' with probability at least $99/100$.
    \end{itemize}
    Then $\calA$ must draw $\Omega({n\choose k})$ samples from $\calD$.
\end{proposition}

Note that any algorithm $\calA$ with the performance guarantee as in \Cref{prop:parity-lb} can distinguish between $\calU_n$ and the (uniform) \emph{mixture} distribution on $\calP_n(k)$ with probability at least $99/100$; in particular, it suffices for us to prove a sample-complexity lower bound against algorithms that distinguish between these two distributions. It thus suffices to prove a sample complexity lower bound against the natural ``maximum likelihood estimate''-based distinguisher. 

\begin{notation}
We will sometimes write $\overline{\bx} := (\x{1},\ldots,\x{T})$ and $\overline{\by} := (\y{1},\ldots,\y{T})$) for brevity.
\end{notation}

We next formally define the MLE-based distinguisher: 

\begin{definition} \label{def:mle-distinguisher}
    Given $T$ i.i.d. samples $(\x{1},\ldots,\x{T})$ drawn from some unknown distribution $\calD \in \{\calU_n\} \cup \calP_n(k)$, we write $\MLE$ for the algorithm that outputs ``truncated'' if
    \[
        \Ex_{\calD \in \calP_{n}(k)}\sbra{\Prx_{\y{1},\ldots,\y{T}\sim\calD}\sbra{\overline{\by} = \overline{\bx}}}\\
        \geq \Prx_{\y{1},\ldots,\y{T}\sim\calU_n}\sbra{\overline{\by} = \overline{\bx}},
    \]
    and ``un-truncated" otherwise.
\end{definition}

\subsection{A Lower Bound Against $\MLE$}
\label{subsec:final-lb}

In the rest of this section, we establish the following proposition which together with \Cref{def:mle-distinguisher} implies \Cref{prop:parity-lb}, which in turn implies \Cref{thm:lb}.

\begin{proposition} \label{prop:mle-goal}
    Given i.i.d. sample access to $\calD$ where 
    \[\calD \in \{\calU_n\} \cup \calP_n(k),\]
    suppose $\MLE$ has the following performance guarantee:
    \begin{itemize}
        \item If $\calD = \calU_n$, then $\MLE$ outputs ``un-truncated'' with probability at least $99/100$; and
        \item If $\calD \in \calP_n(k)$, then $\MLE$ outputs ``truncated'' with probability at least $99/100$.
    \end{itemize}
    Then $\MLE$ must draw at least $\log{n\choose k}$ samples from $\calD$.
\end{proposition}


To prove that $\MLE$ drawing $\log\binom{n}{k}$ samples does not satisfy the performance guarantee, we prove \Cref{lemma:mle-to-weight} which reduces our problem to proving an anticoncentration bound for a coding-theoretic random variable. Finally, \Cref{lemma:DONE} completes the proof by providing the anticoncentration bound needed.

We define our random variable as follows:

\begin{definition} \label{def:codesss}
    Let $\bV \leq \F_2^n$ be a random linear subspace obtained by the following process:
    \begin{enumerate}
        \item Draw $\v{1},\ldots,\v{T} \sim\F_2^n$ uniformly at random where $T := \log{n\choose k}$.
        \item Set $\bV := \mathrm{span}(\v{1},\ldots,\v{T})^\perp$.
    \end{enumerate}
    We also define 
    \[\bw := \abs{\{x \in \bV : |x| = k\}}.\]
\end{definition}

The following lemma relates the failure probability of $\MLE$ to a tail probability of the random variable $\bw$:

\begin{lemma} \label{lemma:mle-to-weight}
    Given $T := \log{n\choose k}$ samples $\x{1},\ldots, \x{T}$ drawn from $\calU_n$, we have 
    \[\Prx_{\calU_n}\sbra{\MLE(\x{1},\ldots,\x{T})~\text{outputs ``truncated''}} = \Prx_{\calU_n}\sbra{\bw \geq 1}\]
    with $\bw$ as in \Cref{def:codesss}.
\end{lemma}

\begin{proof}
    By \Cref{def:mle-distinguisher}, we have that the probability of $
    \MLE$ outputting ``truncated'' given samples from $\calU$ is given by
    \begin{equation} \label{eq:melon}
    { \Prx_{\x{i}\sim\calU_n}\sbra{\Ex_{\calD \in \calP_{n}(k)}{}\sbra{\Prx_{\substack{\y{i}\sim\calD}}\sbra{\overline{\by}} = \overline{\bx}} \geq \Prx_{\y{i}\sim\calU_n}\sbra{\overline{\by}} = \overline{\bx}}.}
    \end{equation}
    Note, however, that for fixed $x^{(1)},\ldots,x^{(T)}$ we have 
    \begin{equation} \label{eq:potato}        {\Prx_{\y{i}\sim\calU_n}\sbra{\overline{\by} = \overline{x}} = \frac{1}{2^{nT}},}
    \end{equation}
    and also that 
    \begin{equation} \label{eq:tomato}
        {\Ex_{\calD \in \calP_{n}(k)} \sbra{\Prx_{\substack{\y{i}\sim\calD}}\sbra{\overline{\by} = \overline{x}}} = \frac{\abs{\cbra{S : |S| = k ~\text{and}~\chi_{S}(x^{(i)}) = 1~\text{for}~i\in[T]}}}{|\calP_n(k)|2^{(n-1)T}}.}
    \end{equation}
    This relies on the fact that $\Vol(\chi_S) = 1/2$  for all $S\neq \emptyset$. Combining \Cref{eq:melon,eq:potato,eq:tomato}, we get that
    \begin{align*}
        \Prx_{\calU_n}\sbra{\MLE~\text{outputs ``truncated''}}
        &=\Prx_{\calU_n}\sbra{\small\abs{\cbra{S : |S| = k ~\text{and for all}~i\in[T],\chi_{S}(\x{i}) = 1}} \geq \frac{|\calP_n(k)|}{2^T}} \\
        &=\Pr[\bw\geq 1]
    \end{align*}
    as $|\calP_n(k)| = {n\choose k}$ and recalling our choice of $T = \log{n\choose k}$.
\end{proof}

With \Cref{lemma:mle-to-weight} in hand, the proof of the lower bound is completed by the following lemma:

\begin{lemma} \label{lemma:DONE}
    We have 
    $\Pr[\bw \geq 1] > 0.01.$
\end{lemma}

\begin{proof}
    We will first establish that $\E[\bw] = 1$. To see this, note that 
    \begin{align}
        \Ex[\bw] &= \sum_{\substack{x\in\F_2^n \\ |x| = k}} \Prx_{\bV}\sbra{x\in\bV} \nonumber \\
        &= \sum_{\substack{x\in\F_2^n \\ |x| = k}} \Prx_{\v{i}\sim\F_2^n}\sbra{\langle \v{i}, x\rangle = 0 ~\text{for all}~i\in[T]} \nonumber\\
        &= \sum_{\substack{x\in\F_2^n \\ |x| = k}} \prod_{i=1}^T\Prx_{\v{i}\sim\F_2^n}\sbra{\langle \v{i}, x\rangle = 0 } \nonumber \\
        &= \sum_{\substack{x\in\F_2^n \\ |x| = k}} \frac{1}{2^T} \nonumber \\
        &= 1, \label{eq:mean-w}
    \end{align}
    recalling that $T = \log{n\choose k}$. We will next obtain an upper bound on the second moment of $\bw$. Recalling that $\bw = \sum_{|x| = k} \mathbf{1}\cbra{x\in\bV}$, we have that 
    \begin{align}
        \Ex[\bw^2] &= \Ex_{\bV}\sbra{\pbra{\sum_{|x| = k} \mathbf{1}\cbra{x\in\bV}}^2} \nonumber\\
        &= \Ex_{\bV}\sbra{{\sum_{|x|, |y| = k} \mathbf{1}\cbra{x\in\bV}\cdot\mathbf{1}\cbra{y\in\bV}}} \nonumber\\
        &= \sum_{\substack{x,y\in\F_2^n \\ |x| = |y| = k}} \Prx_{\bV}\sbra{x \in \bV, y\in \bV} \nonumber\\ 
        &= \sum_{\substack{x\in\F_2^n \\ |x| = k}} \Prx_{\bV}\sbra{x \in \bV} + \sum_{\substack{x\neq y \\ |x| = |y| = k}} \Prx_{\bV}\sbra{x \in \bV, y\in \bV}, \nonumber
        \intertext{but when $x\neq y$, note that the events $\cbra{x\in \bV}$ and $\cbra{y\in \bV}$ are independent, and so we have}
        &= 1 + \sum_{\substack{x\neq y \\ |x| = |y| = k}} \Prx_{\bV}\sbra{x \in \bV}\Prx_{\bV}\sbra{y \in \bV} \nonumber\\
        &\leq 2, \label{eq:moment2-w}
    \end{align}
    where the last two expressions relied on our choice of $T$.
    
    We now turn to the proof of \Cref{lemma:DONE}. Suppose, for the sake of contradiction, that $\Pr[\bw\geq 1] \leq 0.01$. Now, from \Cref{eq:mean-w} we have that 
    \[1 = \E[\bw] = \sum_{i\geq 1}\Prx[\bw \geq i].\]
    By assumption we then have that 
    \begin{equation} \label{eq:to-contradict}
       0.01\geq \Pr[\bw \geq 1]=1-\sum_{i\geq 2}\Pr[\bw\geq i], \qquad\text{and so}\qquad \sum_{i\geq 2}\Pr[\bw\geq i] \geq 0.99. 
    \end{equation}
    Note, however, that by \Cref{eq:mean-w,eq:moment2-w} it follows that $\Var[\bw] \leq 1$ and so by Chebyshev's inequality we have that for any $i \geq 2$,
    \[\Pr[\bw\geq i] \leq \frac{1}{(i-1)^2}.\]
    We also have that for $i\geq 2$, $\Pr[\bw\geq i] \leq \Pr[\bw\geq 1] \leq 0.01$ by assumption, and so 
    \begin{align*}
        \sum_{i\geq 2}\Pr[\bw\geq i] &\leq \sum_{i\geq 2}\min\cbra{\frac{1}{(i-1)^2}, 0.01}\\
        &\leq \sum_{i=2}^{15} 0.01 + \sum_{i > 15} \frac{1}{(i-1)^2}\\
        &\leq 0.15 + 0.25\\
        &\leq 0.4,
    \end{align*}
    which contradicts \Cref{eq:to-contradict}. It follows that $\Pr[\bw\geq 1] > 0.01$, completing the proof.
\end{proof}

\section*{Acknowledgements}
S.N. is supported by NSF grants IIS-1838154, CCF-2106429, CCF-2211238, CCF-1763970, and CCF-2107187. The authors would like to thank Lucas Gretta, Fermi Ma and Avishay Tal for helpful discussions. This work was partially completed while the authors were visiting the Simons Institute for the Theory of Computing. 

\bibliographystyle{alpha}
\bibliography{references}

\end{document}